\documentclass[11pt]{article}
\pdfoutput=1

\usepackage{amsfonts}
\usepackage{amssymb,amsmath,amsthm}
\usepackage{latexsym}
\usepackage{fullpage}
\usepackage{hyperref}
\usepackage{graphicx}
\usepackage{array}
\usepackage{color}

\definecolor{Red}{rgb}{1,0,0}
\definecolor{Blue}{rgb}{0,0,1}
\definecolor{Olive}{rgb}{0.41,0.55,0.13}
\definecolor{Green}{rgb}{0,1,0}
\definecolor{MGreen}{rgb}{0,0.8,0}
\definecolor{DGreen}{rgb}{0,0.55,0}
\definecolor{Yellow}{rgb}{1,1,0}
\definecolor{Cyan}{rgb}{0,1,1}
\definecolor{Magenta}{rgb}{1,0,1}
\definecolor{Orange}{rgb}{1,.5,0}
\definecolor{Violet}{rgb}{.5,0,.5}
\definecolor{Purple}{rgb}{.75,0,.25}
\definecolor{Brown}{rgb}{.65,.4,.25}
\definecolor{Grey}{rgb}{.5,.5,.5}

\newtheorem{theorem}{Theorem}[section]

\newtheorem{cor}[theorem]{Corollary}

\newtheorem{fact}[theorem]{Fact}

\theoremstyle{definition}
\newtheorem{remark}{Remark}[section]

\newcounter{tenumerate}

\def\P{\mathbb{P}}

\newcommand{\one}{\1}

\renewcommand{\epsilon}{\varepsilon}

\newcommand{\1}{\mathbf{1}}

\newcommand{\N}{{\mathbb N}}

\newcommand{\E}{{\mathbb E}}
\newcommand{\remove}[1]{}

\renewcommand{\leq}{\leqslant}
\renewcommand{\geq}{\geqslant}

\def\XXint#1#2#3{{\setbox0=\hbox{$#1{#2#3}{\int}$}
\vcenter{\hbox{$#2#3$}}\kern-.5\wd0}}

\title{The Hitchhiker's Guide to Affiliation Networks:\\A Game-Theoretic Approach}

\author{
Christian Borgs\thanks{Microsoft Research New England. Email: [borgs,jchayes]@microsoft.com}
\and
Jennifer Chayes\footnotemark[1]
\and
Jian Ding\thanks{Department of Statistics, UC Berkeley. Email: jding@stat.berkeley.edu} \footnotemark[4]
\and
Brendan Lucier\thanks{Department of Computer Science, University of Toronto. Email: blucier@cs.toronto.edu} \thanks{This work was done while JD and BL were interns at Microsoft Research New England.}
}

\begin{document}
\setcounter{page}{0}

\maketitle

\abstract{
We propose a new class of game-theoretic models for network formation in which strategies are not directly related to edge choices, but instead correspond more generally to the exertion of social effort.  This differs from existing models in both formulation and results: the observed social network is a byproduct of a more expressive strategic interaction, which can more naturally explain the emergence of complex social structures.  Within this framework, we present a natural network formation game in which agent utilities are locally defined and that, despite its simplicity, nevertheless produces a rich class of equilibria that exhibit structural properties commonly observed in social networks -- such as triadic closure -- that have proved elusive in most existing models.

Specifically, we consider a game in which players organize networking events at a cost that grows with the number of attendees.  An event's cost is assumed by the organizer but the benefit accrues equally to all attendees: a link is formed between any two players who see each other at more than a certain number $r_0$ of events per time period, whether at events organized by themselves or by third parties.  The graph of connections so obtained is the social network of the model.

We analyze the Nash equilibria of this game for the case in which
each player derives a benefit $a>0$ from all her neighbors in the
social network and when the costs are linear,  i.e., when the cost
of an event with $\ell$ invitees is $b+c\ell$, with $b>0$ and $c>0$.
For $\gamma=a/cr_0>1$ and $b$ sufficiently small, all Nash
equilibria have the complete graph as their social network;
for $\gamma<1$ the Nash equilibria correspond to a rich class of
social networks, all of which have substantial clustering in the
sense that the clustering coefficient is bounded below by the
inverse of the average degree.  Many observed social network
structures occur as Nash equilibria of this model.  In particular,
for any degree sequence with finite mean, and not too many vertices
of degree one or two, we can construct a Nash equilibrium producing
a social network with the given degree sequence.


We also briefly discuss generalizations of this model to more complex utility functions
and processes by which the resulting social network is formed.
}

\newpage

\section{Introduction}

In the past decade there has been increasing interest in complex social networks that arise in both
online and offline contexts.  
Empirical studies have found these networks to share many properties, most notably small diameter, heavy-tailed degree distributions, and substantial clustering. 
In light of these observations, theoretical work has focused on explaining \emph{how} and \emph{why} such features appear.
While many of these properties have been modeled probabilistically, it is less well understood why they arise as outcomes of strategic behavior.
%
To this end, we present a new class of game-theoretic models that generate networks with many commonly-observed structural properties which have proved elusive in existing strategic models.

Probabilistic approaches to modeling network formation include static random graph models
(such as the classic random graph model \cite{Bollobas01} and the configuration model \cite{MR93}) and dynamic random graph models (such as the preferential attachment model \cite{BA99} and its many variants).
Many observed properties of social networks have been captured in such models: small diameter is easy to achieve (though often not easy to prove) by the inclusion of randomness \cite{BR04}; heavy-tailed degree distributions arise from many preferential attachment models \cite{BRST01};
and clustering has been established in several random dynamic models such as the copying \cite{KRRSTU00} and affiliation \cite{LS09} models.  On the other hand, these models do not take into account
individual incentives in the development of connections, and therefore do not provide an understanding of how networks
arise from individual choices.  For this, a game-theoretic approach is better suited.

Game-theoretic approaches to network formation go back to the work
of Boorman \cite{Boorman75}, Aumann and Myerson \cite{AM88}, and
Myerson \cite{Myerson91}. For a beautiful treatment,
see the recent books by Goyal \cite{Goyal07}, Jackson \cite{Jackson08}, and 
Easley and Kleinberg \cite{EK10}.
Generally speaking, the strategic models studied
in the literature share a common high-level description: agents are assumed to
derive some benefit from a social network, but connections come
at a cost (be it in the form of effort, money, or otherwise). Agents
thus strategically choose which connections to maintain, either
unilaterally or in pairs, and then reap the benefits of the
resulting network.  One would then predict that observed social
networks correspond to the equilibria of the resulting game.
Although there has been some very thoughtful and insightful work on
strategic models of network formation (see \cite{Goyal07},
\cite{Jackson08}, and \cite{EK10}, and references therein), the
social networks so obtained at equilibrium tend to have rather
limited architectures.  In particular, they tend not to exhibit many
structural properties associated with observed social networks, such
as triadic closure.\footnote{Triadic closure refers to the tendency for
a node's neighbours to connect to each other.}

In this paper, we present a new class of game-theoretic models for network formation, which to our knowledge are
qualitatively different, in both formulation and results, from previous models.  Our
approach is motivated by the sociological notion of \emph{affiliation networks} due to Breiger \cite{Breiger74, Breiger90}.
An affiliation network is a bipartite graph
  where the nodes of one type are the players, and the nodes of the other type are events or groups with which
   the players can be affiliated.  One can also view an affiliation network as a hypergraph over players, with each
    hyperedge representing a single group.  This hypergraph can then be used to
induce a network on the set of players by linking any two that are members of at least one common group.

We build upon the affiliation network model in two ways.  First, we
take a game-theoretic approach whereby each affiliation group is
sponsored by an agent (at a cost), while the benefits from the
resulting network are derived by all players (allowing
 some individuals to hitchhike on the efforts of others).  Second, we change the induced network so that coaffiliation does not
  automatically guarantee a connection: agents must jointly participate in many events in order to form a link.
More concretely, we study a model in which agents hold social
events, where an event with $\ell$ invitees has cost $b+c\ell$ for
$b,c > 0$.   Agents who see each other at more than $r_0$ events
become connected, for some
$r_0\in \N$, and these connections form the observed social network.
Each player then receives a benefit $a > 0$ for each of
 their neighbors in the network.


Formulated more abstractly, our paradigm is that the actions of the players in a network formation game
are not necessarily directly associated with the formation of edges, but rather take place in an underlying strategy space.
 The resulting network
  is determined by an interplay of these actions which is not necessarily just a
  union of edges chosen by the players.
%
%
We will assume that the cost $\mathcal C_v$ for player $v$ is a
function of her strategy $\mathcal P_v$ and the strategic choices
of all agents imply a network $G = G(\{\mathcal P_v\}_{v \in V})$.
The benefit $\mathcal B_v$ of agent $v$ is then a function only of this
social network (as opposed to the full strategy space), leading to a utility
\[
\mathcal U_v=\mathcal B_v(G)-\mathcal C_v(\mathcal P_v)
  \]
for player $v$.  Some classes of examples for benefit functions $\mathcal B_v$ include distance-based benefits, such as those introduced by Bloch and Jackson \cite{BJ07}, as well as  benefits that do not decay with distance, such as those as studied by Bala and Goyal \cite{BG00}.


\paragraph{Results}

Returning to our particular example of a social event game, we study
the equilibria of behaviour as a function of parameter $\gamma =
a/cr_0$.
  When $\gamma > 1$ and $b$ is sufficiently small, all equilibria result in the complete graph (Fact \ref{fact.complete}).  On the other hand,
  for $\gamma < 1$, there is a much richer class of equilibria that all demonstrate a high degree of triadic closure.  Specifically, the induced
  subgraph of the ball of radius $1$ centered at any vertex $v$ can be written as a union of triangles and at most one additional edge
  (Theorem \ref{thm.weak.ties}).  From this we can derive tight lower bounds on the clustering coefficient of any resulting social network
  in terms of the average degree; see Theorem \ref{thm.cluster}.  It is notable that clustering arises endogenously in sparse
  graphs, even though an agent's benefit from the network depends only on her degree.

\begin{figure}
\begin{center}
\begin{tabular}{c}
\includegraphics[width=2in]{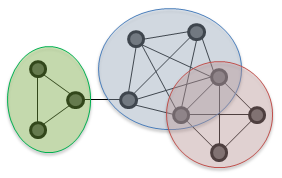} \\
\end{tabular}
\end{center}
\caption{A typical connected component of a network at equilibrium, where cliques (circled for emphasis) are connected arbitrarily by local bridges and/or small overlapping subgroups.}
\label{fig:sample}
\end{figure}

We then consider the larger structural properties of networks
supportable at equilibrium.  We find that such networks are built up
from cliques whose minimal size depends on parameter $\gamma$;
loosely speaking, we find that larger cliques are supportable for
wider ranges of parameters (see Fact \ref{fact.supportable.clique}).
Such cliques can be thought of as tightly-knit groups, which can then
be loosely connected into arbitrary
social structures.  More precisely, given a graph $H$ on $n$
vertices with degrees $d_1, \dotsc, d_n$, and complete graphs $H_1,
\dotsc, H_n$ with $|H_i| \geq \min\{d_i,3\}$, we can construct a
supportable social network in which the vertices of $H$ are replaced
by the cliques $\{H_i\}$ and the edges of $H$ represent local
bridges\footnote{A local bridge is an edge whose endpoints do not
share a common neighbour.} between those communities.  Such a
structure can also be supported by overlapping communities, as long
as the non-overlapping regions are sufficiently large.  See Figure
\ref{fig:sample}, Fact \ref{fact.bowtie}, and Remark \ref{rem.community}.
Moreover, such
structures allow us to support networks with any given heavy-tailed
degree distribution, as in Theorem \ref{thm-degree-sequence}.

\paragraph{Related Work}

The notion of affiliation networks was introduced by Breiger \cite{Breiger74, Breiger90}
and expanded upon by McPherson \cite{McPherson82};
see also Wasserman and Faust \cite{WF89} and references therein.  A theoretical treatment of
this model was given by Lattanzi and Sivakumar \cite{LS09}, who showed that a variant of the
copy model of random network growth applied to affiliation networks leads to graphs that exhibit
heavy-tailed degree distributions, low diameter, and edge densification.  This work differs from
our results in that it does not consider strategic issues.

Game-theoretic models of network formation go back to Boorman \cite{Boorman75} and
Myerson \cite{Myerson91}.
Very roughly speaking, most strategic models studied in the literature can be divided into two classes: those in which links
 are formed by  unilateral decisions, and those in which decisions are made pairwise.  In the case of unilateral
 decisions, the appropriate notion of stability is Nash equilibrium.  Bala and Goyal \cite{BG00} present models of
 unilateral link formation where a player's benefit is the size of the reachable component, either via directed
or undirected edges.  In the former case,
  the efficient equilibria architectures are cycles or empty networks; in the latter, the resulting equilibria are trees.
Fabrikant et.\ al.\ \cite{FKP02} consider a more complex utility model motivated by routing costs, and obtain
equilibria as trees with heavy-tailed degree distributions.

In the case that mutual consent is required to form a link, the
natural notion of stability is that of pairwise stability: no
individual player wants to sever a link and no pair of players wants
to add a link between them.  Many such models
 of network formation have been studied in recent years.  The distance-based utility model of Bloch and Jackson \cite{BJ07},
 as well as variants of the model due to Fabrikant et. al. \cite{FLMPS03}, have either the complete graph or stars as the unique
 resulting efficient networks. The coauthor model of Jackson and Wolinsky \cite{JW95}  has disconnected cliques as its efficient
 resulting network, while the so-called island connection model due to Jackson and Rogers \cite{JR05} generates cliques connected
 by single edges.


\section{A Model of Social Effort}

Let $V$ be a community of rational agents.  We wish to describe a
network formation game through which the agents form connections. We
begin by describing the strategy space. An \emph{event} is a subset
$P \subset V$ of agents, along with a corresponding \emph{rate} $r >
0$. A strategy for an agent $v\in V$ is a sequence of events
$\mathcal{P}_v = \{P_{v, 1}, \ldots, P_{v, k_v}\}$ with
corresponding rates $r_{v,1}, \dots, r_{v, k_v}$.  Here $k_v$ is the
number of events initiated by $v$.  We will refer to a strategy
profile $\mathcal{P} = \{(\mathcal{P}_v, r_v)\ : \ v \in V\}$ as an
\emph{event configuration}.
Fixing an event configuration and individuals $u,v \in V$, the
\emph{meeting rate supported by $w \in V$} between
$u$ and $v$ is
$$M_{u,v}^w = \mbox{$\sum_{i \in [k_w]}$}\,r_{w, i} \one_{\{u, v\in P_{w, i}\}}\,.$$
Thus $M_{u,v}^w$ denotes the rate at which $u$ and $v$ are both present in
an event held by $w$.
The \emph{meeting rate} $M_{u,v}$ is the total rate of the events that both
$u$ and $v$ attend.
That is,
\begin{equation}\label{eq-meet-rate}
M_{u,v} =\mbox{$\sum_{w \in V} $}\,M_{u,v}^w\,.
\end{equation}
We say that agents $u$ and $v$ are \emph{connected} if their meeting
rate is at least some threshold $r_0$.  Without loss of generality
we scale values so that $r_0 = 1$.  We write $N_v$ for the set of
individuals connected to $v$.  That is,
\begin{equation}\label{eq-friends}
N_v = \{u\neq v: M_{u,v} \geq 1\}\,.
\end{equation}
This notion of connectedness is symmetric, so that $u \in N_v$
if and only if $v \in N_u$.

\paragraph{Utility Model} We now describe the utilities in our network
formation game. We assume that an agent obtains a benefit $a>0$ for
each agent to which he is connected. Also, the cost for agent $v$ to
hold events $P_{v,1}, \dotsc, P_{v,k_v}$ at rates $r_{v,1}, \dotsc,
r_{v,k_v}$ is
\begin{equation}\label{eq-party-cost}
\mathcal{C}_v(\mathcal{P}_v) = \mbox{$\sum_{i\in [k_v]}$}\,r_{v,i}(
c(|P_{v,i}\backslash\{v\}| + b)
\end{equation}
where $c>0$ counts the cost of an event per agent (excluding $v$
herself) and $b> 0$ counts the initial fixed cost of an event.
Altogether, given event configuration $\mathcal{P}$, the utility of
agent $v$ is
\begin{equation}\label{eq-utility}
\mathcal{U}_v = a|N_v| - \mathcal{C}_v(\mathcal{P}_v) = a|N_v| -
\mbox{$\sum_{i\in [k_v]}$} \,r_{v, i} (c|P_{v, i}\backslash\{v\}| + b)\,.
\end{equation}
We will write $\gamma = \frac{a}{c}$ throughout.

\paragraph{Stability} We say an event configuration is
\emph{stable} if it forms a Nash equilibrium in the network formation
game.
An event configuration naturally defines a network, where we place
an edge between $u$ and $v$ precisely when $u \in N_v$.
We say that this network of connections is \emph{supported} by
event configuration $\mathcal{P}$.
We say that a graph is
\emph{supportable} if there exists at least one stable event
configuration that supports it.


\section{Characterization of Stable Networks}

We now wish to characterize networks that arise as Nash equilibria
of our network formation game.  We begin by considering the optimal
response of an agent $v$ given the strategy profile of the other
agents.
We say that the event strategy $\mathcal{P}_v$ \emph{realizes}
invitation rates $\mathcal{M}_v = \{m_{v, u}^v: u\neq v\}$ if $M_{v,
u}^v = m_{v, u}^v$ for all  $u\neq v$. A key ingredient in analyzing
optimal responses for $v$ is to show how to realize a given
$\mathcal{M}^v$ with an event configuration of minimal cost.

\begin{theorem}\label{thm-best-response}
Given $\mathcal{M}^v = \{m_{v, u}^v: u\neq v\}$, any optimal
strategy $\mathcal{P}_v$ realizing $\mathcal{M}^v$ satisfies
\begin{equation}\label{eq-best-response}\mbox{$\sum_{i\in [k_v]}$}\, r_{v, i} = \max_{u\neq v} m_{v,
u}^v\,.\end{equation}
\end{theorem}

\begin{proof}
We first demonstrate the existence of a strategy that realizes
$\mathcal{M}^v$ and satisfies \eqref{eq-best-response}. Take  all
the strictly positive $m_{v, u}^v$'s and arrange them in a
decreasing order such that $m_{v, v_1}^v \geq \ldots\geq m_{v,
v_\ell} >0$. Let $P_{v, i} = \{ v, v_1, \ldots, v_i\}$ for $i=1,
\ldots, \ell$. We define the event strategy
$\mathcal{P}_v^{\mathrm{nest}}$ to be the collection of events
$\{P_{v, 1}, \ldots, P_{v, \ell}\}$ at corresponding rates
$$r_{v, i} = m_{v, v_i}^v - m_{v, v_{i+1}}^v \mbox{ for } i = 1, \ldots, \ell,$$
where we used the convention $m_{v, v_{\ell+1}}^v = 0$. It is
straightforward to verify that the strategy
$\mathcal{P}_v^{\mathrm{nest}}$ realizes $\mathcal{M}^v$ and also
satisfies
\eqref{eq-best-response} (note that in this case $k_v = \ell$).

Now, let $\mathcal{P}_v$ be an arbitrary strategy that realizes
$\mathcal{M}^v$ but violates \eqref{eq-best-response}. Observe that
$\mathcal{P}_v$ must satisfy $\sum_{i=1}^{k_v} r_{v, i} \geq
\max_{u\neq v} m_{v, u}^v$. Then, assuming \eqref{eq-best-response}
is violated, we  have $\sum_{i=1}^{k_v} r_{v, i} > \max_{u\neq v}
m_{v, u}^v$. Therefore, we conclude that
\begin{align*}
\mathcal{C}_v(\mathcal{P}_v)& = \sum_{i=1}^{k_v} r_{v, i} (c |P_{v,
i}| + b) = b \sum_{i=1}^{k_v} r_{v, i} + c \sum_{u\neq v} r_{v, i}
\one_{\{u\in P_{v, i}\}} = b \sum_{i=1}^{k_v} r_{v, i} + c
\sum_{u\neq v} m_{v, u}^v\\
&> b \max_{u\neq v}m_{v, u}^v + c \sum_{u\neq v} m_{v, u}^v =
\mathcal{C}_v(\mathcal{P}_v^{\mathrm{nest}})\,.\qedhere
\end{align*}
\end{proof}


\begin{cor}
Any stable strategy profile $\mathcal{P}$ satisfies that $
\sum_{i=1}^{k_v} r_{v, i} = \max_{u\neq v} M_{v, u}^v$, for $v\in
V$.
\end{cor}

\begin{cor}
\label{cor.argmax}
If $\mathcal{P}$ is a stable event configuration and
$u \in \arg\max_{u\neq v}\{M_{v,u}^v\}$ then $u \in P_{v,i}$
for all $1 \leq i \leq k_v$.
\end{cor}

We next aim at a set of criteria for an event configuration
$\mathcal{P}$ to be stable. Given agent $v$ and strategies
$\mathcal{P}\backslash\mathcal{P}_{v}$ of the other agents, define
$E_{v,u} =  1 - \sum_{w\neq v}M_{v, u}^w$. We think of $E_{v,u}$ as
the minimal rate at which $v$ should invite $u$ in order to create a
connection. Let $T_v = \{u\in V: 0< E_{v,u} < \gamma\}$. Finally,
given an event profile $\mathcal{P}_v$ for agent $v$, let $I_v =
\{u\in V: M_{v,u}^v> 0 \}$ denote the set of \emph{invitees} for
$v$.
\begin{theorem}
\label{thm.stable}
An event configuration $\mathcal{P}$ is stable
if and only if, for all $v\in V$,
\begin{equation}\label{eq-condition}
I_v \subseteq T_v \,, \quad M_{v,u}^v = E_{v,u} \,, \mbox{ and } \
E_{v,u} < E_{v,w} \mbox{ for all } u\in I_v \mbox{ and } w\in
T_v\setminus I_v.
\end{equation}
\begin{proof}
We demonstrate the required conditions in order. Keep in mind that
the quantities $E_{v,u}$ completely capture the impact of a given
strategy $\mathcal{P}_v$ on the utility of $v$. First of all, $v$
should only hold events that include agents $u \in T_v$, since for
any other agent the marginal utility of supporting the connection is
negative. Also, there is no point for $v$ to realize an invitation
rate $M_{v, u}^v
> E_{v, u}$ for any agent $u$, since $v$ and $u$ will be connected
as long as $M_{v, u}^v = E_{v, u}$ and making $M_{v, u}^v$ larger
will only cost more to $v$. Similarly, there is no point to make
$0<M_{v, u}^v < E_{v, u}$, since it will incur a cost but offers no
benefit. Furthermore, for $x, y\in T_v$ such that $E_{v, x} \geq
E_{v, y}$, if $v$ profits by making a connection with $x$, it must
profit by doing so with $y$ (note that it may not be optimal for $v$
to support connections with \emph{all} agents in $T_v$, due to the
fixed cost component $b$ in the utility model).

\end{proof}
\end{theorem}

\begin{cor}
\label{cor.delta.bound}
In any stable event configuration, $M_{v,u}^v < \gamma$ for all $u,v \in V$,
and $M_{v,u}^v > 0$ implies $u \in N_v$.
\end{cor}




\section{Properties of Supportable Networks}

We now wish to analyze the properties of networks that are supportable
by stable event configurations.  We begin by considering simple examples, in
order to build some intuition for the structures that can arise in stable
networks.  We then consider the clustering coefficient and average degree of
supportable networks.

\subsection{Examples}

In this section we give some simple examples of network structures, along with
necessary conditions for them to be supportable.  We begin by noting that if $\gamma > 1$
and $b$ is sufficiently small, then the complete graph is the only supportable graph.

\begin{fact}
\label{fact.complete}
If $\gamma > 1$ and $b < c(\gamma-1)$ then $K_n$ is the unique supportable graph.
\end{fact}

Given the above result, we will focus on the case $\gamma < 1$.
In what follows we will also suppose that $b = 0^+$ is set to be arbitrarily small.


\begin{table}
\begin{center}
\begin{tabular}{|c|c|c|}
\hline
\includegraphics[width=0.42in, viewport = 0 0 80 80]{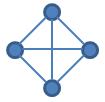} &
\includegraphics[width=0.4in]{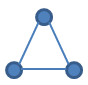} &
\includegraphics[width=0.45in]{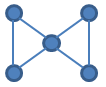} \\
$\gamma > 1/4$ & $\gamma > 1/3$ & $1/3 < \gamma < 1$ \\
\hline \hline
\includegraphics[width=0.5in]{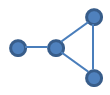} &
\includegraphics[width=0.5in, viewport = 0 0 80 80]{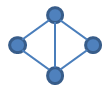} &
\includegraphics[width=0.4in]{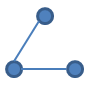} \\
$1/2 \leq \gamma < 1$ &
$1/2 \leq \gamma < 1$ &
Not supportable \\
\hline
\end{tabular}
\end{center}
\caption{Sample connection graphs, with the range of parameter $\gamma$ in which they are supportable as strong subgraphs.}
\label{fig:examples}
\end{table}




We say that $H$ is a \emph{strong subgraph} of network $G$ if $H$ is a subgraph of $G$, and moreover, for each $u,v \in H$, $N_u \cap N_v \subseteq H$.  We consider several examples of small graphs and study when they can be strong subgraphs of a supportable graph.

\begin{fact}
\label{fact.supportable.clique} Graph $K_\ell$ can be a strong
subgraph of a supportable network if and only if $\gamma >
\frac{1}{\ell}$.
\end{fact}

An important special case is 
an edge $(u,v)$ with $N_u \cap N_v = \emptyset$, which we call a \emph{local bridge}.  We show that each node is incident with at most one local bridge in a supportable network.


\begin{theorem}
\label{thm.weak.ties}
A supportable graph $G$ can contain $K_2$ as a strong subgraph only if $\gamma > 1/2$.  Moreover, each node can be contained in at most one such subgraph.
\end{theorem}
\begin{proof}
The condition on $\gamma$ follows from Fact \ref{fact.supportable.clique}.  Next suppose for contradiction that $G$ is supported by a stable invitation graph and node $u$ is connected to multiple agents with whom he shares no common neighbours.  Choose $x \in \arg\max\{M_{u,x}^u\}$.  Then, in particular, there is some $v \neq x$ such that $v \in N_u$ and $u$ and $v$ do not share any common neighbours.  Corollary \ref{cor.argmax} then implies that $x \in P_{u,i}$ for all $1 \leq i \leq k_u$, and in particular for each $i$ in which $v \in P_{u,i}$.  We therefore have $E_{v,x} \leq 1 - M_{u,v}^u \leq M_{v,u}^v = E_{v,u}$, which contradicts Theorem \ref{thm.stable}.
\end{proof}

\begin{cor}
The star graph with more than $1$ leaf is not supportable.
\end{cor}

For any $k > p \geq 1$, we will write $H_{k,p}$ for the graph on $2k - p$ vertices which consists of two $k$-cliques which share
$p$ vertices in common.  
Theorem \ref{thm.weak.ties} can be re-interpreted as demonstrating
that graph $H_{2,1}$ is not supportable.  We next demonstrate that,
for any $k$, $H_{k,p}$ is supportable if $k-p$ is not too small
(depending on $\gamma$).

\begin{fact}
\label{fact.bowtie}
For $k > 2$, a supportable graph can contain $H_{k,1}$ as a strong subgraph if and only if $\gamma > 1/k$.
For $p > 1$ and $k > p+1$, a supportable graph can contain $H_{k,p}$ as a strong subgraph if $\gamma > 1/(k-p)$.
\end{fact}

\begin{remark}
\label{rem.community}
Fact \ref{fact.bowtie} generalizes to allow arbitrarily many cliques of varying sizes to be joined together at single vertices
(or by bridges if $\gamma > 1/2$) or overlap with sufficiently small intersections.
It is therefore possible to build networks in which tightly-knit communities (i.e. cliques)
are joined by small intersections and/or loose inter-community connections.
\end{remark}



We have shown that cliques can be joined at a single vertex without altering the threshold at which they are supportable.  We conclude our list of examples by observing that cliques joined at multiple vertices can have a different supportability threshold on $\gamma$ (and, in particular, are harder to support as strong subgraphs).

\begin{theorem}
\label{thm.h32}
Graph $H_{3,2}$ is supportable if and only if $\gamma > 1/2$.
\end{theorem}

\subsection{Clustering in Supportable Networks}

One phenomenon observed in social community is that two people with a common neighbor tend to be connected with each other.
So, one expects to see many triangles in the connection graph.  A common measure of clustering is the
\emph{clustering coefficient} of graph $G$, defined as
\begin{equation}\label{eq-def-cluster}
\mathcal{E}(G) = \frac{1}{|\{v\in V: d_v \geq 2\}|}\sum_{v \in V: d_v\geq 2} \frac{N_{v,\triangle}}{\binom{d_v}{2}}\,,
\end{equation}
where $d_v$ is the degree of $v$ and $N_{v, \triangle}$ is the number of triangles to which $v$ belongs. That is to say, if we pick a random node in the graph of degree at least 2 and pick two random neighbors of this node, the probability that these 3 nodes form a triangle is $\mathcal{E}(G)$.  Additionally, we write $\bar{d}_G$ for the average degree of $G$.

\begin{theorem}
\label{thm.cluster}
Any connected supportable connection graph $G$ satisfies $\mathcal{E}(G)
\geq 1/(2\bar{d}_{G})$.
\end{theorem}
\begin{remark}
For any graph $G$ with minimal degree $2$, we have $\mathcal{E}(G)
\geq 1/\bar{d}_{G}$.
\end{remark}
\begin{proof}
Let $W = \{v\in V: d_v \geq 2\}$.
For every $v\in W$, Theorem \ref{thm.weak.ties} asserts that $v$ can have at most one friend $u$ such that $N_v \cap N_u = \emptyset$.
This implies that $N_{v, \triangle} \geq \frac{d_v-1}{2}$. Therefore,
\begin{equation}\label{eq-clustering}
\mathcal{E}(G) \geq \frac{1}{|W|}\sum_{v\in W} \frac{(d_v-1)/2}{\binom{d_v}{2}} = \frac{1}{|W|}\sum_{v\in W} \frac{1}{d_v} \geq \frac{1}{\bar{d}_{G, W}}\,,\end{equation}
where $\bar{d}_{G, W}$ is the average degree over set $W$ and the last transition uses the convexity of the function $f(x)
= 1/x$ for $x >0$.

Note that in a connected supportable graph, any degree 1 vertex has
to be connected to a vertex of degree at least 2. On the other hand,
Theorem~\ref{thm.weak.ties} implies that any vertex can be connected
to at most one vertex of degree 1. Altogether, we see that
$|V\setminus W| \leq |V|/2$. Therefore, $\bar{d}_G \geq \bar{d}_{G,
W}/2$. Combined with \eqref{eq-clustering}, the required bound
follows.
\end{proof}

Note that our bound on the clustering coefficient is tight, up to a
constant depending on $\gamma$.

\begin{theorem}\label{thm.cluster.upper.bound}
Suppose $1/k<\gamma <1-1/k$ for some $k\in \N$. Then there exist
supportable connection graphs for which $\mathcal{E}(G) \leq
\frac{(1+o(1))(k-2)}{\bar{d}_G-1}$.
\end{theorem}
\begin{proof}
Construct a random $d$-regular $k$-hypergraph $H$ on $n$ nodes. Let
$m = dn/k$ and denote by $\mathcal{X} = \{X_1, \ldots, X_m\}$ the
random hyperedges in $H$. Define
$$\mathcal{Y} = \{X_i: |X_i \cap X_j|\geq 2 \mbox{ for some } i\neq j\} \cup \{X_i: X_i \mbox{ is in a cycle of length at most } 4 \}\,.$$
It is straightforward to compute that for all $i\in [m]$
$$\P(X_i \in \mathcal{Y}) \leq m \big(\tfrac{k}{n}\big)^2 + \tbinom{m}{2}\big(\tfrac{k^2}{n}\big)^3 + \tbinom{m}{3} \big(\tfrac{k^2}{n}\big)^4 = o(1)\,.$$
It follows that with high probability $|\mathcal{Y}| = o(m)$. Let
$H'$ be the hypergraph with edge set $\mathcal{X}\setminus
\mathcal{Y}$. We now construct an event configuration in the
following way: for each hyperedge $A\in H$, let everyone in $A$ host
an event inviting everyone else in $A$ with rate $1/k$. We next
demonstrate that this configuration is stable. Note that each
hyperedge generates a clique with meeting rate 1, and since $\gamma
> 1/k$ no agent is motivated to reduce these meeting rates.  For two
agents who are not in any same hyperedge, their meeting rate is at
most $1/k$ since the girth of $H'$ is at least $5$.  Since $\gamma <
1-\frac{1}{k}$, no agent is motivated to make an additional
connection since it would require an invitation rate $1-1/k >
\gamma$. This completes the verification of stability.

Recalling that $|\mathcal{Y}| = o(m)$, we see that every vertex $v
\in G$ except for $o(n)$ vertices is in $d$ cliques of size $k$,
where these cliques are disjoint except for the intersection at $v$.
That is to say, for a $(1-o(1))$ fraction of the vertices, we have
$d_v(G) = (k-1)d$ and $N_{v, \triangle}(G) = d\binom{k-1}{2}$.
Therefore, we conclude that
\begin{equation*}\mathcal{E}(G)\geq (1+o(1))(k-2)/((k-1)d-1) \geq (1+o(1))(k-2)/(\bar{d}_G -1)\,.\qedhere\end{equation*}
\end{proof}

\subsection{Event Size and Network Sparsity}

We note that our network formation model does not impose any limits
on the size of the events that agents can support.  Indeed, it is
possible for a single agent to hold an event for all agents in the
network.  However, we note that such large events are not necessary
to obtain a lower bound on clustering (Theorem \ref{thm.cluster}) or
support interesting degree distributions (see Section
\ref{sec.powerlaw}). Thus, even though our strategy space allows for
very large events, we obtain equilibria in which each agent holds
only small events.

One could argue that configurations with small events are
natural, since in many settings it seems unlikely that a single agent would
unilaterally support a large fraction of an entire social network.
Given that such small-event configurations arise as equilibria in our model,
we turn to studying their properties.

Recall that $I_v = \{u\in V: M_{v,u}^v> 0 \}$ is the set of
individuals invited to events held by $v$. We say that a connection
graph is $K$-supportable if it is supportable by an event
configuration in which $|I_v| \leq K$ for each aget $v$. We give an
upper bound on the average degree of $K$-supportable connection
networks. Combined with Theorem \ref{thm.cluster}, it yields a lower
bound of $1/(2\gamma K(K+1))$ on clustering coefficient for any
$K$-supportable connected graph.
\begin{theorem}
Any $K$-supportable connection graph $G$ satisfies $\bar{d}_G \leq
\gamma K(K+1)$.
\end{theorem}
\begin{proof}
For the proof, we consider the corresponding weighted connection
graph $\tilde{G}$ with edge weight $w_{u, v} = M_{u, v}$. It is
clear from our definition that $\bar{d}_G \leq \bar{d}_{\tilde{G}}$.
Note that in every stable party configuration, each agent $v$ can
only invite at most $K$ people with rate at most $\gamma$.
Therefore, its invitation can contribute at most $\gamma K(K+1)$ to
the total degree of $\tilde{G}$. Summing over all the agents, we get
$\bar{d}_{\tilde{G}} \leq \gamma K (K+1)$.
\end{proof}

We also note that the average degree of network $G$ can
indeed approach the bound of $\gamma K^2$.

\begin{theorem}\label{thm.degree.lower.bound}
Suppose that $K>\lfloor 1/\gamma \rfloor + 1$. There exists a
$K$-supportable connection graph such that $\bar{d}_G =
(1-o(1))K(K+1)/(\lfloor 1/\gamma \rfloor + 1)$.
\end{theorem}


\section{Supportable Degree Sequences}
\label{sec.powerlaw}

 We show in this
subsection that for a rich family of degree sequences, there exists
a corresponding supportable connection network. In
particular, we demonstrate that we can support a connected
 graph of power law degree with finite mean.
\begin{theorem}\label{thm-degree-sequence}
Fix $1/2 < \gamma < 2/3$. Let $D = \{d_1, \ldots, d_n\} \in [n]^n$ be a degree sequence such that
\begin{enumerate}
\item There exists $K \in \N$ such that $|\{i \in [n]: 2\leq d_i\leq K\}| \geq \sum_{i\in[n]} d_i\one_{\{d_i >
K\}}$.
\item $|\{i\in [n]: d_i = 2\}| \leq \sum_{i\in[n]}(d_i-3) \one_{\{d_i\geq
6\}}$.
\item $|\{i\in [n]: d_i = 1\}| \leq \frac{1}{3}|\{i \in [n]: d_i \geq 4\}|$.
\end{enumerate}
Then for some constant $C_K \geq 0$ depending on $K$, there exists a
connected $(K+3)$-supportable graph $G$ of degree sequence $D' =
\{d'_1, \ldots, d'_n\}$ such that the $\ell_1$ shift satisfies
$$\|D - D'\|_{\ell_1} = \mbox{$\sum_{i\in [n]}$}\, |d_i - d'_i| \leq C_K\,.$$
\end{theorem}

\begin{proof}
Let $A = \{i\in [n]: d_i\leq K\}$ and $B = \{i\in n: d_i
>K\}$. Our construction consists of the following several steps.
Keep in mind that we can perturb the degree sequence by some amount
depending on $K$ and we use this fact throughout the proof. For simplicity, we first state the construction assuming there are
no degree 1 vertices at all, and we will address this issue later. After the construction, we will discuss the connectedness and stability of the graph.

\noindent{\bf Step 1: Handling degree-2 vertices.} By Assumption
(2), we can effectively remove  degree-2 vertices by adding triangles to
high-degree nodes. Precisely, we repeat the following procedure.
Pick $u, v\in A$ of degree 2 and $w =
\arg\max_{x\in A\cup B}d_x$. Let $u$ hold events $\{v, w\}$ and $v$
hold events $\{u, w\}$, both at rate $1/2$. This supports a
stable triangle among $\{u, v, w\}$. Remove $u, v$ from $A$ and
update $d_w$ to $d_w - 2$. We stop the process when there is at
most 1 degree 2 vertex left, at which point we perturb its degree a bit and
make it 3. Note that Assumption (1) is preserved.

\smallskip

\noindent{\bf Step 2: Handling high-degree vertices.} Step 1 allows us to assume that $\min_i d_i\geq 3$. In this step, we
can further reduce to the case where $3 \leq d_i\leq K+3 $ for all $i$.
We can find a degree sequence $D^*$ with $\|D^* -
D\|_{\ell_1} \leq K^2$ such that Assumption (1) is preserved and
$|\{i\in [n]: d_i = k\}|$ is a multiple of $k$ for all $3\leq k\leq K$.
We then repeatedly match high degree vertices while preserving
Assumptions (1). If there is $v\in B$ such that $d_v \geq K+3$,
by Assumption (1) there exist $k$ vertices of degree $k$ in
$A$ for some $3\leq k\leq K$. Then, we will form a clique containing these $k$
vertices together with $v$ as follows: each of these $k$ vertices
holds an event inviting $v$ and all the other $k-1$ vertices at rate $1/k$.
We now remove these $k$ vertices from $A$ and update $d_v$ as $d_v - k$. We
remark that the number of vertices consumed from set $A$ is $k$ and
the total degree consumed from $B$ is exactly $k$ and
therefore Assumption (1) is preserved. This justifies that we can
repeat this process until $\max_v d_v\leq K+3$.

\smallskip

\noindent{\bf Step 3: Constructing regular graphs of low degree.} We can now assume $3\leq d_i\leq K+3$ for all $i$. For every $3\leq k\leq K+3$, we would like to construct a connected
graph that contains all the vertices of degree $k$. Denote by $m_k$ the number of vertices of degree $k$. We
need to use the fact that for every $k$, there exists a connected
$k$-regular graph $H$ on $m$ vertices of girth at least $5$, where
$m\geq n_k$ for some $n_k\in \N$ depending only on $k$.

\noindent{\bf Case 1: $m_k\leq k n_k$.} We first group all the
vertices in blocks of size $(k+1)$ and for each block we support a
clique such that each agent holds an event for
his block with rate $1/(k+1)$. It remains to make the graph connected.
To this end, we pick two vertices $u_j, v_j$ from the $j$-th clique
for every $j$. Now we add edges of form $(u_j, v_{j+1})$ for every
$j$, where $u_j$ and $v_{j+1}$ both hold events for each other
with rate $1/2$. This graph is supportable and
connected and furthermore, the $\ell_1$ shift of the degree sequence is
at most $2 n_k$.

\noindent{\bf Case 2: $m_k\geq k n_k$.} In this case, take a connected
$k$-regular graph $H$ on $m_k/k$ vertices of girth at least $5$. We
construct our connection graph based on $H$. Basically, we replace
each node of $H$ by a clique of size $k$ where each node in the clique holds an event for the whole clique with rate $1/k$. Then, for every edge in
$H$, we take a vertex from each of the two cliques corresponding to
the nodes of this edge, and add an edge between this two vertices by
letting them invite each other with rate $1/2$. We do this in a way
such that we add exactly one edge to every vertex (using the fact that $H$ is $k$-regular). Since the girth
of $H$ is at least 5, we see that in our construction, those who are
not connected have meeting rate at most $1/k \leq 1/3$. Recalling
$\gamma<2/3$, we verify that this construction is stable.

\smallskip

\noindent{\bf Step 4: Connecting graph components.} It remains to
connect these (roughly) $k$-regular graphs constructed in Step 3. We would like to add edges between these graphs (to connect them),
but we must guarantee that no vertex is involved in 2 such bridges (including those used in Step 3). For each
 ``regular'' graph  constructed in Case 1 in the previous step, it is clear
that we can choose two vertices such that there are no bridges associated with them. For graphs constructed in Case 2, since the underlying graph $H$
contains cycles, we can break a bridge between two vertices (without affecting the
connectedness). In either
case, we get two vertices for each such graph with no bridges.
Label these vertices as $v_k, u_k$ for $3\leq k\leq K+3$ and add edges
of the form $(v_k, u_{k+1})$ by letting them invite each other at rate
$1/2$. We now obtain a connected graph and it is stable.

\smallskip

\noindent{\bf Attaching degree-1 vertices.} We now turn to dealing with degree 1 vertices. Basically, up to the availability of degree 1 vertices, we want to modify our construction slightly such we are allowed to attach degree 1 vertices. This modification happens in Steps 2 and 3.

Suppose now we have at least $K$ degree 1 vertices. In Step 2, whenever we construct a clique of size $k$, we can attach degree 1 vertices to the clique we added. More precisely, instead of taking $k$ vertices of degree $k$ from $A$, we now take $k-1$ vertices of degree $k$ from $A$ and denote them by $v_1, \ldots, v_{k-1}$.  We then form a clique for $\{v, v_1, \ldots, v_{k-1}\}$ by having each $v_i$ holds an event for the clique at rate $1/(k-1)$. We then update $d_v$ to $d_v - (k-1)$. If $k\geq 4$, we can now attach a vertex of degree 1 to each $v_i$ where both this newly added vertex and $v_i$ invite each other at rate $1/2$. Remove $\{v_1, \ldots, v_{k-1}\}$ and these attached degree 1 vertices from $A$. We see that Assumptions (1) and (3) are preserved.

In Step (3), we consider every $4\leq k \leq K+3$. If $m_k \leq k n_k$, we can attach a degree 1 vertex to every vertex in the graph that does not have a bridge (of which there are at least $m_k - 2n_k \geq \frac{m_k}{3}$) and this perturbs the degree of at most $m_k$ vertices by 1. If $m_k>kn_k$, we consider the $k$-regular connection graph we constructed. We see that we can remove up to $|H|(k-2)/2$ bridges and the graph will remain connected. Now, for each bridge that is broken, we can attach a degree 1 vertex to each end of this bridge using mutual invitation rate $1/2$. This implies that we can add as many as $(k-2)|H| = \frac{k-1}{k} m_k \geq \frac{m_k}{3}$ degree 1 vertices in this step.  This completes the consideration of degree 1 vertices.

\smallskip

We are done with the construction. It is easy to see from our 4
steps that the total $\ell_1$ shift of the degree sequence is
bounded by some $C_K>0$ depending only on $K$. In Step 1 and 2, we
removed many low degree vertices and every low degree vertex is
connected to some vertex which remains. In Step 3 and 4, we
constructed a connected graph out of all the remaining vertices.
This implies that the final connection graph is connected. Note also
that at each step of the construction, those vertices who remain
have not held any event, so there is no interplay between steps that
would affect the stability of the construction. That is to say, the
stability of each step (as demonstrated above) implies the stability
of the whole construction. Finally, the total number of people an
agent invites is bounded by $K+3$. This completes the proof.
\end{proof}

\begin{cor}
Fix $1/2<\gamma<2/3$. Let  $D = \{d_1, \ldots,
d_n\} \in [n]^n$ be a degree sequence such that
\begin{itemize}
\item $|\{i \in [n]: d_i = k\}| =  \lfloor c n k^{-\alpha} \rfloor$, for some constant $c>0$ and $\alpha>2$.
\item $|\{i\in [n]: d_i = 2\}| \leq \sum_{i\in[n]}(d_i-3) \one_{\{d_i\geq
6\}}$.
\item $|\{i\in [n]: d_i = 1\}| \leq \frac{1}{3}|\{i \in [n]: d_i \geq 4\}|$.
\end{itemize}
Then, there exist connected $K$-supportable graphs of degree
sequence $D' = \{d'_1, \ldots, d'_n\}$ such that $\|D -
D'\|_{\ell_1} = \sum_{i\in [n]} |d_i - d'_i| \leq C_\alpha$, where
$K, C_\alpha>0$ depend only on $\alpha$.
\end{cor}
\begin{proof}
We only need to verify Assumption (1) in Theorem
\ref{thm-degree-sequence}.  For power law degree with $\alpha>2$,
the average degree is finite and thus Assumption (1) holds for
some $K_\alpha$ depending on $\alpha$. Applying Theorem
\ref{thm-degree-sequence} completes the proof.
\end{proof}

\begin{remark} We note that if we impose a lower bound of $d$ on the minimal degree, we can apply a similar construction and yield a relaxed condition on $\gamma$; namely, that $1/2 < \gamma < 1 - \frac{1}{d}$.
\end{remark}

\section{Future Directions}

The strategic model of affiliation networks presented herein
represents a step toward the larger goal of obtaining a
game-theoretic understanding of the structural properties of social
networks.  As such, there are many natural extensions to consider
and questions to pose.

In this work we have focused on the static properties of networks at equilibrium.  A natural next step is to determine which equilibria are likely to arise as outcomes when a network evolves over time.  As a particular example, one might consider a growth model in which new agents arrive and initiate certain events, after which point the network attempts to stabilize via best-response dynamics.  What are the properties of networks that form according to such a process?

Our model extends easily to allow heterogeneity between agents.  In particular, one could allow connection benefits and event costs to vary between individuals (or pairs of individuals).  Such a generalization could, for example, distinguish agents according to a measure of sociability.  One could also model the effects of event size by having the meeting rate of individuals depend on the sizes of the events that they co-attend.  This could be used to model the fact that agents are less likely to build a relationship at a large gala than at an intimate dinner party.

We have focused on the solution concept of Nash equilibrium, where
deviations are assumed to occur unilaterally.
 One could extend this to allow for multiple agents to deviate jointly.  For example, one might imagine a bargaining dynamics by
 which agents jointly decide to sponser social groups, lowering the barriers to link formation.  This could be viewed as an extension
 of Nash Bargaining to groups of agents. Alternatively, one might impose a cost on attending events, which adds a cooperative
 element to the game: agents need not only propose events, but also choose whether or not to accept invitations.  Such a model
 harkens to the notion of pairwise stability that has been studied extensively in the strategic network formation literature.

\appendix

\section*{Appendix A}
\begin{proof}[\emph{\textbf{Proof of Fact \ref{fact.complete}}}]
We first note that $K_n$ is supportable by a stable event
configuration in which an agent $v$ holds event $P = V$ with rate
$1$.  Next suppose $G = (V,E)$ and there exist $u,v \in V$ with
$(u,v) \not\in E$. The utility of node $u$ would increase by at
least $a-c-b$ if he were to hold an event for $v$ with rate $1$,
which is strictly positive if $b < c(\gamma-1)$.  Thus $G$ is not
supportable.
\end{proof}

\begin{proof}[\emph{\textbf{Proof of Fact \ref{fact.supportable.clique}}}]
Let $H$ be a strongly connected subgraph of size $\ell$.  If $\gamma
\leq \frac{1}{\ell}$, then by Corollary \ref{cor.delta.bound}
$M_{v,u}^v < \frac{1}{\ell}$ for all $u,v \in H$.  Thus, it must be
that $M_{v,u} \leq \sum_{w \in V}M_{u,v}^w < \ell \cdot
\frac{1}{\ell} = 1$ for all $u,v \in H$.  Thus $(u,v) \not\in H$,
and hence $K_\ell$ is not supportable.  On the other hand, if
$\gamma > \frac{1}{\ell}$ then graph $K_\ell$ can be supported by
the event configuration in which each agent invites all others to a
single event with rate $\frac{1}{\ell}$.
\end{proof}

\begin{proof}[\emph{\textbf{Proof of Fact \ref{fact.bowtie}}}]
Let $v$ denote the single node in $H_{k,1}$ that connects the two
cliques of size $k$.  If $1/k < \gamma < 1-1/k$, graph $H_{k,1}$ can
be supported by having each node invite all of his neighbours to a
single event with rate $1/k$.
 If $\gamma \geq 1-1/k$, we can instead support $H_{k,1}$ by having each node except $v$ invite all of his neighbours
 to a single event with rate $1/(k-1)$, and $v$ holds no events.  In either case, graph $H_{k,1}$ is supportable. On the
 other hand, if $\gamma \leq 1/k$, then it must be that $M_{v,u}^v < 1/k$ for all $u,v \in H_{k,1}$.  Thus, since each
 pair of nodes in $H_{k,1}$ have at most $k-2$ common neighbours, it follows from Corollary \ref{cor.delta.bound} that
 $M_{v,u} < 1$ for all $u,v \in H_{k,1}$, and thus $H_{k,1}$ is not supportable.

For the more general case of $H_{k,p}$, let $A$ and $B$ denote the
two disjoint cliques of size $k-p$, and let $C$ denote the single
clique of size $p$. Then $H_{k,p}$ can be supported by the event
structure in which each $u \in A$ holds an event for $A \cup C$ at
rate $1/(k-p)$
 and each $v \in B$ holds an event for $B \cup C$ at rate $1/(k-p)$.
 \end{proof}

\begin{proof}[\emph{\textbf{Proof of Theorem \ref{thm.h32}}}]
Let $G = H_{3,2}$.  See Figure \ref{fig:h32}(a) for an illustration
of this graph, with a labeling of the vertices.  We show that $G$ is
supportable by explicitly giving a stable event configuration.  The
invitation rates of this event configuration are illustrated in
Figure \ref{fig:h32}(b), with the convention that a label on an edge
$(v_1,v_2)$ near vertex $v_1$ is $M_{v_1,v_2}^{v_1}$ (i.e.\ the rate
at which $v_1$ invites $v_2$).  Thus, for example, we have
$M_{x,v}^x = \frac{1-\gamma}{2}$ and $M_{w,v}^w = \frac{1}{2}$.
These invitation rates are realized through the configuration that
nests events, as in the proof of Theorem \ref{thm-best-response}. We
encourage the reader to verify that this event structure is indeed
stable when $\gamma > 1/2$, and that it supports graph $H_{3,2}$.

\begin{figure}
\begin{center}
\begin{tabular}{c|c|c}
\includegraphics[width=1.7in]{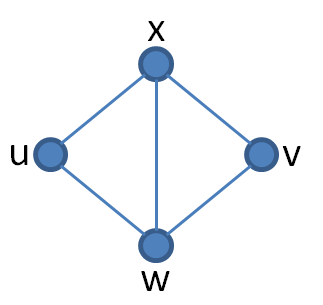} & \includegraphics[width=1.7in]{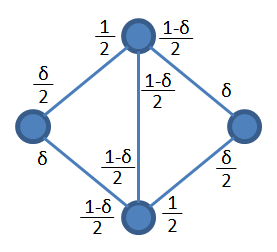} & \includegraphics[width=1.7in]{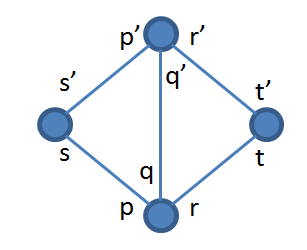}\\
(a) & (b) & (c)
\end{tabular}
\end{center}
\caption{Graph $H_{3,2}$ from the proof of Theorem \ref{thm.h32}. (a) A labeling of the vertices. (b) The invitation rates of a supporting event configuration, with the convention that a label on edge $(v_1, v_2)$ near node $v_1$ represents $M_{v_1,v_2}^{v_1}$. (c) A labeling of the invitation rates used in the proof of Theorem \ref{thm.h32}. }
\label{fig:h32}
\end{figure}

We will now show that $H_{3,2}$ is not supportable when $\gamma \leq 1/2$.  Suppose for contradiction that graph $H_{3,2}$ is supportable by a stable event configuration.  We will show that $M_{u,v} > 1-\gamma$, which contradicts Theorem \ref{thm.stable} (since node $u$ would have incentive to support a connection to node $v$).  For ease of exposition, we will assign a variable to each of the ten different invitation rates in our graph; see Figure \ref{fig:h32}(c) for this labeling.  So, for example, we will write $q'$ for $M_{x,w}^x$.

Note first that, since $\gamma \leq 1/2$, it cannot be that any invitation rate is $0$.  This is because each of the four outer edges has only one common neighbour, so if any invitation rate is $0$ then some edge is being supported by only two agents, which by Corollary \ref{cor.delta.bound} is not possible if $\gamma \leq 1/2$.  It must therefore be that each invitation rate is precisely equal to the effort required to maintain the associated connection.  In particular,
\begin{equation}
\label{eq-qq'}
M_{x,w} = \min\{s,s'\} + \min\{t,t'\} + q + q' = 1.
\end{equation}

We wish to give a lower bound on $M_{u,v}$, but this is complicated by the fact that $M_{u,v}^x$ and $M_{u,v}^w$ depend on the way that nodes $x$ and $w$ realize their invitation rates.  We therefore consider separate cases depending on the values of the invitation rates for nodes $x$ and $w$.
\textbf{Case 1: $q > p$ and $q > r$.}  Note that $M_{w,v}^x \leq q'$, and thus $q' + r + t \geq M_{w,v} = 1$.  If it were the case that $t \leq t'$, we would conclude
\[1 \leq q' + t + r < q' + t + q \leq q' + q + \min\{s,s'\} + \min\{t,t'\} = 1\,,\]
contradicting \eqref{eq-qq'}.  We must therefore have $t > t'$.  By symmetry we can also conclude that $s > s'$.  In the same way, if $q' > p'$ or $q' > r'$ we would conclude $s' > s$ or $t' > t$, a contradiction.  It must therefore be that $q' \leq p'$ and $q' \leq r'$.

By symmetry we can assume $p' \leq r'$.  We then claim that $M_{u,v} \geq p' + (p+r-q)$.  This follows from Corollary \ref{cor.argmax}: node $v$ must be included in all events held by node $x$, and hence $M_{u,v}^x=p'$; and node $x$ must be part of all events held by node $w$, so nodes $u$ and $v$ must be together in such events with rate at least $p+r-q$, by inclusion-exclusion principle.  Furthermore, since $s > s'$ and $t > t'$, \eqref{eq-qq'} implies $q+q'+s'+t' = 1$.  Finally, $q > p$ and $q > r$ implies $p + s' + p' = M_{x,u} = 1$ and $r + t' + r' = M_{x,v} = 1$.  Putting this all together, we conclude $$M_{u,v} \geq p'+(p+r-q) = 1-(r'-q') > 1-r' > 1-\gamma,$$ as required, since each invitation rate is less than $\gamma$.

\textbf{Case 2: $q' > p'$ and $q' > r'$.}  This implies $M_{u,v} > 1-\gamma$ in the same way as Case 1.

\textbf{Case 3: $q \leq \max\{p, r\}$ and $q' \leq \max\{p', r'\}$.} We note that $q+s'+p' \geq M_{x,u} = 1$ and $q'+r+t \geq M_{w,v} = 1$, which implies $q+q'+s'+t > 1$ since $p' < \gamma \leq 1/2$ and $r < \gamma \leq 1/2$.  Equation \eqref{eq-qq'} therefore implies that we cannot have $s' \leq s$ and $t \leq t'$.  Similarly, it cannot be that $s \leq s'$ and $t' \leq t$.  We can therefore assume by symmetry that $s < s'$ and $t < t'$, so that \eqref{eq-qq'} implies $q+q'+s+t = 1$.

We conclude by considering cases for $p$, $r$, $p'$, and $q'$.  Suppose that $p \leq r$ and $p' \leq r'$. Recalling that $M_{u,v} = \min\{p,r\} + \min\{p',r'\}$, we then have $M_{u, v} = p+p'$.
Since
\begin{equation}
\label{eq.m.uw}
q' + p + s \geq M_{u,w} = 1
\end{equation}
and
\begin{equation}
\label{eq.m.ux}
q + p' + s' \geq M_{u,x} = 1,
\end{equation}
we obtain
\[ 2 \leq q' + q + p + p' + s + s' = q + q' + s + s' + M_{u, v}\,. \]
Combined with $q+q'+s+t = 1$ and $s' <\gamma$, it follows that $M_{u, v} > 1-\gamma$ as required.
The cases for $p > r$ and/or $p' > r'$ are handled similarly, applying inequalities $q' + r + t \geq 1$ in place of \eqref{eq.m.uw} and $q + r' + t' \geq 1$ in place of \eqref{eq.m.ux} as appropriate.
\end{proof}

\begin{proof}[\emph{\textbf{Proof of Theorem \ref{thm.degree.lower.bound}}}]
Let $n = |V|$ be the total community size. Write $k^\star =\lfloor
1/\gamma \rfloor + 1$. Split $V$ into groups of $k^\star$ agents and
denote these groups by $A_1, \ldots, A_m$, where $m = n/k^\star$.
Now take i.i.d\ random subsets $X_i \subset V$ of size
$K+1-\lceil1/\gamma\rceil$ from the community. Define
$$Y_k = \bigcup_{i<k}\left\{v\in X_k: v\in X_i, (A_k \cup Y_k)\cap X_i \neq \emptyset \right\} \mbox{ and } Z_k = X_k \setminus Y_k\,.$$
That is, the subsets $Z_i$ are similar to the subsets $X_i$, but
``fixed'' so that no pair of agents appears in two different
subsets. Now let each node in $A_k$ invite $A_k \cup Z_k$ (except
itself) with rate $\gamma$. By our definition of $Z_k$, the meeting
rate for any two agents is either $0$ or $1$. This verifies that the
configuration is stable. We now count the total degree $D$ for the
corresponding connection graph. Recalling that the meeting rate is
either 0 or 1, we see $D = \sum_{k=1}^m
(|Z_k|+k^\star-1)(|Z_k|+k^\star)$. Note that, by the union bound,
$$\E Y_k \leq \gamma m \frac{K+1}{n} \frac{(K+1)^2}{n} = o(1)\,.$$
This implies that with high probability, we have  $|k\in[n]: Y_k
\geq 1| = o(m)$. It then follows that with high probability $$D\geq
(1-o(1))m K(K+1) = (1 -o(1)) n K(K+1) /k^\star\,.$$ This guarantees
the existence of a $K$-supportable graph with the required average
degree.
\end{proof}

\bibliography{socialnetwork}

\begin{thebibliography}{10}

\bibitem{AM88}
R.~J. Aumann and R.~B. Myerson.
\newblock Endogenous formation of links between players and of coalitions: an
  application of the shapley value.
\newblock In A.~Roth, editor, {\em In the shapley value}. Cambridge University
  Press, 1998.

\bibitem{BG00}
V.~Bala and S.~Goyal.
\newblock A noncooperative model of network formation.
\newblock {\em Econometrica}, pages 1181--1229, 2000.

\bibitem{BA99}
A.-L. Barab�si and R.~Albert.
\newblock Emergence of scaling in random networks.
\newblock {\em Science}, pages 509--512, 1999.

\bibitem{BJ07}
F.~Bloch and M.~O. Jackson.
\newblock The formation of networks with transfers among players.
\newblock {\em Journal of Economic Theory}, 133(1):83--110, March 2007.

\bibitem{Bollobas01}
B.~Bollob{\'a}s.
\newblock {\em Random graphs}, volume~73 of {\em Cambridge Studies in Advanced
  Mathematics}.
\newblock Cambridge University Press, Cambridge, second edition, 2001.

\bibitem{BR04}
B.~Bollob{\'a}s and O.~Riordan.
\newblock The diameter of a scale-free random graph.
\newblock {\em Combinatorica}, 24(1):5--34, 2004.

\bibitem{BRST01}
B.~Bollob{\'a}s, O.~Riordan, J.~Spencer, and G.~Tusn{\'a}dy.
\newblock The degree sequence of a scale-free random graph process.
\newblock {\em Random Structures Algorithms}, 18(3):279--290, 2001.

\bibitem{Boorman75}
S.~A. Boorman.
\newblock A combinatorial optimization model for transmission of job
  information through contact networks.
\newblock {\em Bell Journal of Economics}, 6(1):216--249, 1975.

\bibitem{Breiger74}
R.~L. Breiger.
\newblock The duality of persons and groups.
\newblock {\em Social Forces}, 53(2):181--190, 1974.

\bibitem{Breiger90}
R.~L. Breiger.
\newblock Social control and social networks: A model from georg simmel.
\newblock In C.~Calhoun, M.~Meyer, and W.~Scott, editors, {\em Structures of
  power and constraint: papers in honor of Peter M. Blau}, pages 453--476.
  Cambridge University Press, 1990.

\bibitem{EK10}
D.~Easley and J.~Kleinberg.
\newblock {\em Networks, Crowds, and Markets: Reasoning About a Highly
  Connected World}.
\newblock Cambridge University Press, 2010.

\bibitem{FKP02}
A.~Fabrikant, E.~Koutsoupias, and C.~H. Papadimitriou.
\newblock Heuristically optimized trade-offs: A new paradigm for power laws in
  the internet.
\newblock In {\em ICALP '02: Proceedings of the 29th International Colloquium
  on Automata, Languages and Programming}, 2002.

\bibitem{FLMPS03}
A.~Fabrikant, A.~Luthra, E.~Maneva, C.~H. Papadimitriou, and S.~Shenker.
\newblock On a network creation game.
\newblock In {\em PODC '03: Proceedings of the twenty-second annual symposium
  on Principles of distributed computing}, 2003.

\bibitem{Goyal07}
S.~Goyal.
\newblock {\em Connections: An introduction to the economics of networks}.
\newblock Princeton University Press, 2007.

\bibitem{Jackson08}
M.~O. Jackson.
\newblock {\em Social and Economic Networks}.
\newblock Princeton University Press, 2008.

\bibitem{JR05}
M.~O. Jackson and B.~W. Rogers.
\newblock The economics of small worlds.
\newblock Game theory and information, EconWPA, Mar. 2005.

\bibitem{JW95}
M.~O. Jackson and A.~Wolinsky.
\newblock A strategic model of social and economic networks.
\newblock In {\em CMSEMS Discussion Paper 1098, Northwestern University,
  revised}, 1995.

\bibitem{KRRSTU00}
R.~Kumar, P.~Raghavan, S.~Rajagopalan, D.~Sivakumar, A.~Tomkins, and E.~Upfal.
\newblock Stochastic models for the web graph.
\newblock In {\em FOCS '00: Proceedings of the 41st Annual Symposium on
  Foundations of Computer Science}, 2000.

\bibitem{LS09}
S.~Lattanzi and D.~Sivakumar.
\newblock Affiliation networks.
\newblock In {\em STOC '09: Proceedings of the 41st annual ACM symposium on
  Theory of computing}, pages 427--434, 2009.

\bibitem{McPherson82}
J.~McPherson.
\newblock Hypernetwork sampling: Duality and differentiation among voluntary
  organizations.
\newblock {\em Social Networks}, 3:225--249, 1982.

\bibitem{MR93}
M.~Molloy and B.~Reed.
\newblock A critical point for random graphs with a given degree sequence.
\newblock In {\em Proceedings of the {S}ixth {I}nternational {S}eminar on
  {R}andom {G}raphs and {P}robabilistic {M}ethods in {C}ombinatorics and
  {C}omputer {S}cience, ``{R}andom {G}raphs '93'' ({P}ozna\'n, 1993)}.

\bibitem{Myerson91}
R.~B. Myerson.
\newblock {\em Game theory: Analysis of conflict}.
\newblock Harvard Univ. Press, 1991.

\bibitem{WF89}
S.~Wasserman and K.~Faust.
\newblock Canonical analysis of the composition and structure of social
  networks.
\newblock In C.~Clogg, editor, {\em Sociological Methodology}, pages 1�42),,
  1989.

\end{thebibliography}
\bibliographystyle{abbrv}

\end{document}